\documentclass[anon]{l4dc2026}

\title[Koopman-BoxQP]{Koopman-BoxQP: Solving Large-Scale NMPC at kHz Rates}

\usepackage{times}
\usepackage{amsmath}
\usepackage{scalerel}
\usepackage{tikz}
\usepackage{algorithm}
\usepackage{enumitem}
\usepackage{multirow}
\usepackage{booktabs}
\usepackage{wrapfig}
\usepackage{subcaption} 
 \usepackage{url}
\usepackage{enumitem}
 \usepackage{todonotes} 

\newcommand*{\pdot}{\mathbin{\scalerel*{\boldsymbol\odot}{\circ}}}

\newcommand{\rr}{{\mathbb R}}

\author{%
 \Name{Liang Wu$^1$} \Email{wliang14@jh.edu}\\
 \Name{Wallace Gian Yion Tan$^2$} \Email{wtgy@mit.edu}\\
 \Name{Richard D. Braatz$^2$} \Email{braatz@mit.edu}\\
 \Name{J\'an Drgo\v na$^1$} \Email{jdrgona1 @jh.edu}\\
 \addr $^1$Johns Hopkins University, MD 21218, USA\\
 \addr $^2$Massachusetts Institute of Technology, MA 02139, USA.
}

\begin{document}

\maketitle

\begin{abstract}
Solving large-scale nonlinear model predictive control (NMPC) problems at kilohertz (kHz) rates on standard processors remains a formidable challenge. This paper proposes a \textit{Koopman-BoxQP} framework that \textit{i)} learns a linear Koopman high-dimensional model, \textit{ii)} eliminates the high-dimensional observables to construct a multi-step prediction model of the states and control inputs, \textit{iii)} penalizes the multi-step prediction model into the objective, which results in a structured box-constrained quadratic program (BoxQP) whose decision variables include both the system states and control inputs, \textit{iv)} develops a \textit{structure-exploited} and \textit{warm-starting-supported} variant of the \textit{feasible} Mehrotra's interior-point algorithm for BoxQP. Numerical results demonstrate that \textit{Koopman-BoxQP} can solve a large-scale NMPC problem with $1040$ variables and $2080$ inequalities at a kHz rate.
\end{abstract}


\section{Introduction}
Modern quadratic programming (QP) solvers  \citep{domahidi2012efficient,ferreau2014qpoases,frison2020hpipm,zanelli2020forces,wu2023simple} are able to solve small-scale linear MPC problems ($<$50 variables and $<$100 inequalities) at kHz rates (milliseconds) on standard processors, but fail to achieve similar performance for medium- or large-scale linear MPC problems (around $1000+$ variables and $2000+$ inequalities), let alone nonlinear programming-based nonlinear MPC problems (NMPC) \citep{diehl2009efficient}. Recently, data-driven Koopman approaches, which learn linear dynamics in a high-dimensional observable space, have emerged as a computationally efficient way to transform NMPC problems into compact general QP formulations by eliminating the high-dimensional observables \citep{korda2018linear,Folkestad2020}. 

While existing Koopman-MPC frameworks use the \textit{condensed} MPC-to-QP construction \citep{korda2018linear} that only involves the control inputs (resulting in fewer decision variables), the number of state inequality constraints can become very large when the state dimension is high, as is commonly the case in control of partial differential equations (PDEs)  \citep{Hassan2018}. As the computation times of QP solvers typically depend on the \textit{total number of variables and inequalities}, solving large-scale Koopman MPC problems at kHz rates using general-purpose QP solvers without exploiting problem structure remains computationally infeasible. Moreover, for practical use, soft-constrained MPC formulations \citep{zeilinger2014soft} that introduce slack variables for state constraints to guarantee feasibility (at the expense of higher computational cost due to the increased number of variables) are commonly employed to handle unknown disturbances and modeling errors. Modeling errors are inevitable, as data-driven finite-dimensional Koopman frameworks always involve projection and sampling errors \citep{zhang2022robust,strasser2025kernel}.

This article proposes a dynamics-relaxed alternative for Koopman-MPC problems, which penalizes the approximated Koopman dynamic model via an $\ell_2$-penalty term in the objective, rather than embedding it into the inequality constraints as in \textit{condensed} Koopman's MPC-to-QP construction \citep{korda2018linear}. The proposed dynamics-relaxed Koopman-MPC approach can be regarded as a \textit{sparse} Koopman's MPC-to-QP construction, as the states are retained, while eliminating the high-dimensional observables in the offline phase. Unlike the \textit{condensed} formulation in \cite{korda2018linear}, which yields a general QP, the proposed dynamics-relaxed Koopman-MPC formulation results in a structured box-constrained QP (BoxQP). To tackle this problem efficiently, this article develops a \textit{structure-exploited} and \textit{warm-starting-supported} (note that the warm-starting does not work in IPMs for general QPs) Mehrotra's predictor–corrector interior point method (IPM) \citep{mehrotra1992implementation}, which is capable of solving the resulting large-scale BoxQP at kHz rates. The overall framework is referred to as the \textit{Koopman-BoxQP} solution for NMPC problems.

\subsection{Contributions}
This article makes the following contributions, which are also the key factors as to how the proposed structure-exploited \textit{Koopman-BoxQP} achieves kHz-rate performance:
\begin{enumerate}[itemsep=0pt,topsep=2pt,leftmargin=12pt]
    \item[1)] \textit{strictly feasible} cold- and warm-start initializations are provided cost-free, enabling direct use of the \textit{feasible} Mehrotra’s predictor–corrector IPM algorithm;
    \item[2)] the \textit{feasible} Mehrotra IPM typically converges to high-accuracy solutions of BoxQPs around $\mathbf{10}$ iterations (under strictly feasible cold-starts), regardless of the problem dimension, which is the key numerical finding highlighted in this article. Moreover, warm-starting performs well in the \textit{feasible} Mehrotra IPM for solving BoxQPs and in real-time MPC problems; warm-starting from previous solutions halves the iterations to about $\mathbf{6}$;
    \item[3)] by exploiting structure, the dimension of the linear system solved in each IPM step is reduced from $5(n_u+n_x)N_p$ to $n_u N_p$, yielding substantial speedups (when $n_x \gg n_u$, as in PDE control).
\end{enumerate}

\section{Problem Formulation}
\begin{wrapfigure}[12]{r}{0.53\textwidth}  
\vspace{-10pt}
\rule{\linewidth}{0.6pt}
\textbf{Nonlinear MPC:}
\begin{subequations}\label{eqn_nmpc}
\begin{align}
\min & \sum_{k=1}^{N-1} 
\|x_{k+1}-x_r\|_{W_x}^2 + \|u_{k}-u_r\|_{W_{u}}^2 \label{subeq:nmpc-cost}\\
& \qquad+ \|u_k-u_{k-1}\|_{W_{\Delta u}}^2 \nonumber\\
\text{s.t. }&x_0 = x(t), \; u_{-1}=0, \\[-2pt]
& x_{k+1} = f(x_k,u_k),\quad k=0,\dots,N-1 \label{subeq:nmpc-dyn} \\[-2pt]
&  -\mathbf{1} \leq x_{k+1}  \leq \mathbf{1},\quad k=0,\dots,N-1 \label{subeq:mpc-state-constraints} \\[-2pt]
& -\mathbf{1}\leq u_{k} \leq \mathbf{1},\quad k=0,\dots,N-1
\end{align}
\end{subequations}
\rule{\linewidth}{0.6pt}
\vspace{-2pt}
\end{wrapfigure}
 This article considers a nonlinear MPC problem (NMPC) for tracking, as shown in \eqref{eqn_nmpc}, where $x(t)$ is the feedback state at the sampling time $t$,  $u_k\in\rr^{n_u}$ and $x_k\in\rr^{n_x}$ denote the control input and state at the $k$th time step, the prediction horizon length is $N$, $x_r$ and $u_r$ denote the desired tracking reference signal for the state and control input, respectively, and $W_x\succ0$, $W_u\succ0$, and $W_{\Delta u}\succ0$ are weighting matrices for the deviation of the state tracking error, control input tracking error, and control input increment, respectively. Assume that $W_x$, $W_u$, and $W_{\Delta u}$ are diagonal. The equality constraint \eqref{subeq:nmpc-dyn} represents the nonlinear discrete-time dynamical system. Without loss of generality, the state and control input constraints, $[x_{\min},x_{\max}]$ and $[u_{\min},u_{\max}]$, are scaled to $[-\mathbf{1},\mathbf{1}]$.
 
NMPC \eqref{eqn_nmpc} is a nonconvex nonlinear program, which can be infeasible, can be sensitive to the noise-contaminated $x(t)$, can have slow solving speed, and can lack of convergence guarantees. A Koopman framework for MPC \citep{korda2018linear}, which transforms NMPC \eqref{eqn_nmpc} into a convex QP problem via data-driven Koopman approximations, allows the use of computationally efficient and convergent QP algorithms for real-time applications.

\section{Preliminary: Koopman transforms NMPC into general QP}
The Koopman operator \citep{koopman1931hamiltonian,koopman1932dynamical} provides a globally linear representation of nonlinear dynamics, enabling efficient linear techniques for nonlinear systems. Originally developed for autonomous systems, Koopman operator theory has since been extended to controlled systems of the form \eqref{subeq:nmpc-dyn} through various schemes \citep{williams2016extending,proctor2018generalizing,korda2018linear}. In practice, the infinite-dimensional Koopman operator is truncated and approximated using data-driven Extended Dynamic Mode Decomposition (EDMD) methods \citep{williams2015data,williams2016extending,korda2018convergence,korda2018linear}. In EDMD specifically, the set of extended observables is designed as the ``lifted" mapping, $\left[\begin{array}{@{}c@{}}
        \psi(x) \\
        \mathbf{u}(0)
    \end{array}\right]\!,$
where $\mathbf{u}(0)$ denotes the first component of the sequence $\mathbf{u}$ and $\psi(x)\triangleq\left[\psi_1(x), \cdots{}, \psi_{n_\psi}(x)\right]^\top$ ($n_\psi \gg n_x$) is chosen from a basis function, e.g., Radial Basis Functions (RBFs) used in \cite{korda2018linear}, instead of directly solving for them via optimization. In particular, the approximate Koopman operator identification problem is reduced to a least-squares problem, which assumes that the sampled data $\{(
x_j,\mathbf{u}_j),(x_j^+,\mathbf{u}_j^+)\}$ (where $j$ denotes the index of data samples and the superscript $+$ denotes the value at the next time step) are collected with the update mapping $\left[\begin{array}{@{}c@{}}
    x_j^+ \\
    \mathbf{u}_j^+
\end{array}\right]\! =\! \left[\begin{array}{@{}c@{}}
     f(x_j,\mathbf{u}_j(0))  \\
     \boldsymbol{S}\mathbf{u}_j
\end{array}\right]$. Then an approximation of the Koopman operator, $\mathcal{A}\triangleq [A \;\;B]$, can be obtained by solving
\begin{equation}
  J(A,B) = \min_{A,B} \sum_{j=1}^{N_d}\|\psi(x_j^+)-A\psi(x_j)-B\mathbf{u}_j(0)\|_2^2.  
\end{equation}
According to \cite{korda2018linear}, if the designed lifted mapping $\psi(x)$ contains the state $x$ after the re-ordering $\psi(x)\leftarrow [x^\top, \psi(x)]^\top$, then $C=[I,0]$.
The learned linear Koopman predictor model is given as
\begin{equation}\label{eqn_Koopman_linear}
\psi_{k+1} = A \psi_k + Bu_k,~ x_{k+1} = C \psi_{k+1},
\end{equation}
where $\psi_k\triangleq\psi(x_k) \in\rr^{n_\psi}$ denotes the lifted state space and with $\psi_0 = \psi(x(t))$. If~\eqref{eqn_Koopman_linear} has a high lifted dimension (for a good approximation), its use in MPC will not increase the dimension of the resulting QP if the high-dimensional observables $\psi_k$ are eliminated via
\begin{equation}\label{eqn_X_U_E_F}
\footnotesize
\left[\begin{array}{@{}c@{}}
             x_{1}  \\
             x_{2} \\
             \vdots \\
             x_{N}
\end{array}\right] = \mathbf{E} \psi\big(x(t)\big) + \mathbf{F} \!\left[\begin{array}{@{}c@{}}
     u_{0}  \\
     u_{1} \\
     \vdots \\
     u_{N-1}
\end{array}\right]\!,~\text{where } \mathbf{E} \triangleq\! \left[\begin{array}{c}
    CA  \\
    CA^2 \\
    \vdots \\
    CA^N
\end{array}\right]\!,\ \  \mathbf{F} \triangleq \left[\begin{array}{cccc}
     CB & 0 & \cdots & 0  \\
     CAB & CB & \cdots & 0 \\
     \vdots & \vdots & \ddots & \vdots \\
     CA^{N-1}B &  CA^{N-2}B & \cdots & CB
\end{array}\right]\!.    
\end{equation}
Then, by embedding \eqref{eqn_X_U_E_F} into the quadratic objective \eqref{subeq:nmpc-cost} and the state constraint \eqref{subeq:mpc-state-constraints}, NMPC \eqref{eqn_nmpc} can be reduced to a compact \textit{general} QP with the decision vector $z\triangleq \mathrm{col}(u_0,\cdots{},u_{N-1})\in\rr^{N\times n_u}$ as shown in \eqref{eqn_KoopmanMPC_QP}, where $\bar{W}_x\triangleq\mathrm{blkdiag}(W_x,\cdots{},W_x)$, $\bar{W}_u\triangleq\mathrm{blkdiag}(W_u,\cdots{},W_u)$, $\bar{R}\triangleq\mathrm{blkdiag}(\bar{W}_{\Delta u},\cdots{},\bar{W}_{\Delta u},\bar{W}_{\Delta u}^N)$ $\Bigg(\footnotesize\bar{W}_{\Delta u}=\left[\begin{array}{cc}
    2W_{\Delta u} &  -W_{\Delta u} \\
    -W_{\Delta u} &  2W_{\Delta u}
\end{array}\right]$ and $\footnotesize\bar{W}_{\Delta u}^N=\left[\begin{array}{cc}
    2W_{\Delta u} &  -W_{\Delta u} \\
    -W_{\Delta u} &  W_{\Delta u}
\end{array}\right]\Bigg)$ and $h \triangleq \mathbf{F}^\top \bar{W}_x(\mathbf{E}\psi(x(t))-\bar{x}_r)-\bar{W}_u\bar{u}_r$ ($\bar{x}_r=\mathrm{repmat}(x_r)$ and $\bar{u}_r=\mathrm{repmat}(u_r)$).
Note that the computation of the high-dimensional observable $\psi(x(t))$ is performed once and will not be involved in the iterations of QP, which minimizes a side effect of the high-dimensional Koopman operator. 
\begin{wrapfigure}[10]{r}{0.52\textwidth} 
\rule{\linewidth}{0.6pt}
\textbf{NMPC $\rightarrow$ Koopman-QP:}
\begin{subequations}\label{eqn_KoopmanMPC_QP}
\begin{align}
\min_{z} &~ z^\top\!\left(\mathbf{F}^\top \bar{W}_x \mathbf{F}+ \bar{W}_u + \bar{R} \right) \!z + 2z^\top h \label{eqn_KoopmanMPC_QP_obj}\\
\text{s.t.} &~ -\mathbf{1}\leq \mathbf{E}\psi(x_t)+\mathbf{F}z \leq \mathbf{1} \label{eqn__KoopmanMPC_QP_state_constraint} \\
&~ - \mathbf{1} \leq z \leq \mathbf{1}
\end{align}
\end{subequations}  
\rule{\linewidth}{0.6pt}
\end{wrapfigure}

The resulting Koopman-QP \eqref{eqn_KoopmanMPC_QP} has $N n_u$ decision variables and $2N (n_x+n_u)$ inequality constraints, with no dependence on $n_\psi$. The Koopman-QP \eqref{eqn_KoopmanMPC_QP} may be infeasible due to \textit{inevitable} modeling errors in the data-driven Koopman \textit{approximation}. By softening the state constraints \eqref{eqn__KoopmanMPC_QP_state_constraint} to $ -\mathbf{1}- \epsilon V_{\min}\leq \mathbf{E}\psi(x_t)+\mathbf{F}z \leq \mathbf{1} +\epsilon V_{\max}$ (the vectors $V_{\min}$ and $V_{\max}$ are non-positive) and adding a penalty term $\rho_{\epsilon} \epsilon^2$ ($\rho_{\epsilon}\gg W_{x},W_{u},W_{\Delta u}$) to the objective \eqref{eqn_KoopmanMPC_QP_obj}, then the Koopman-QP \eqref{eqn_KoopmanMPC_QP} with the new decision vector $z\leftarrow \mathrm{col}(z,\epsilon)$ is guaranteed to be feasible. 
Moreover, the Koopman-QP \eqref{eqn_KoopmanMPC_QP} constructed via condensing is a general QP without a specific structure, yet no existing work has developed a tailored structure-exploited QP solver for the Koopman-MPC framework to enable faster real-time computation.

\section{Methodology}
To simultaneously handle potential infeasibility and obtain faster execution time, this article proposes a \textit{dynamics-relaxed} construction that transforms the Koopman-MPC problem into a parametric BoxQP, which is not always feasible but is guaranteed to be Lipschitz. The most important reason for choosing the construction is that we can customize a \textit{structure-exploited} and \textit{warm-starting-supported} IPM-based QP solver to solve the QP at kHz rates, which general QPs cannot do, as demonstrated in Subsection \ref{sec_IPM_algorithm}.

\subsection{Koopman-BoxQP: Dynamics-relaxed Construction}
\begin{wrapfigure}[9]{r}{0.52\textwidth} 
\vspace{-10pt}
\rule{\linewidth}{0.6pt}
\textbf{Nonlinear MPC $\rightarrow$ \textit{Koopman-BoxQP}}
{\small
\begin{equation}\label{eqn_KoopmanMPC_BoxQP}
\begin{aligned}
\min_{U,X} &~ (X-\bar{x}_r)^\top \bar{W}_x(X-\bar{x}_r) + (U-\bar{u}_r)^\top \bar{W}_u(U-\bar{u}_r)\\
&~  + U^\top \bar{R} U + \rho \|X- \mathbf{E} \psi\big(x(t)\big) - \mathbf{F}U\|_2^2  \\
\text{s.t.} &~ -\mathbf{1} \leq U \leq \mathbf{1}, \\
&~ -\mathbf{1} \leq X \leq \mathbf{1}
\end{aligned}
\end{equation}
}
\rule{\linewidth}{0.6pt}
\end{wrapfigure}
This article proposes an alternative, the \textit{dynamic relaxation approach}, which does not strictly enforce the Koopman prediction model \eqref{eqn_X_U_E_F} but instead incorporates it into the objective through a penalty term, which results in a strongly convex BoxQP  \eqref{eqn_KoopmanMPC_BoxQP}, where $U\triangleq\mathrm{col}(u_0,\cdots{},u_{N-1})$, $X\triangleq\mathrm{col}(x_1,\cdots{},x_{N})$, and $\rho>0$ is a large penalty parameter reflecting the confidence in the Koopman model’s accuracy. 

\begin{wrapfigure}[5]{r}{0.38\textwidth} 
\vspace{-20pt}
\begin{equation}\label{eqn_BoxQP}
    \begin{aligned}
        \min_{z}&~\frac{1}{2}z^\top H z + z^\top h(x(t))\\
        \text{s.t.}&~ -\mathbf{1}\leq z \leq \mathbf{1}        
    \end{aligned}
\end{equation}
\end{wrapfigure}
For simplicity, we denote the decision vector $z\triangleq\mathrm{col}(U,X)\in\rr^n$ ($n=N (n_u+n_x)$), and the proposed \textit{Koopman-BoxQP} \eqref{eqn_KoopmanMPC_BoxQP} is constructed as shown in \eqref{eqn_BoxQP}, where
\begin{equation}\label{eqn_H_def}
H\triangleq \rho \! \left[\begin{array}{@{}cc@{}}
    \mathbf{F}^\top \mathbf{F} & -\mathbf{F}^\top \\
    -\mathbf{F} & I 
\end{array} \right] + \left[\begin{array}{@{}cc@{}}
    \bar{W}_u+\bar{R}&  \\
     & \bar{W}_x
\end{array}\right]\!\succ0, \quad h(x(t))\triangleq \rho\!\left[\begin{array}{@{}c@{}}
      \mathbf{F}^\top\mathbf{E}  \\
      -\mathbf{E} 
 \end{array}\right]\! \psi(x(t)) - \left[\begin{array}{@{}c@{}}
      \bar{W}_u\bar{u}_r  \\
      \bar{W}_x\bar{x}_r 
 \end{array}\right].    
\end{equation}
The dynamics-relaxed \textit{Koopman-BoxQP} \eqref{eqn_KoopmanMPC_BoxQP} can be viewed as an alternative approach to softening the state constraints, since the error in satisfying the prediction model can be equivalently interpreted as a relaxation of the state constraints. 

The dynamics-relaxed \textit{Koopman-BoxQP} \eqref{eqn_KoopmanMPC_BoxQP} is a parametric optimization, which is proved below to be feasible and have a Lipschitz-continuous feedback policy. The former property prevents solver failure, and the latter property avoids chattering or abrupt control actions. Future work will leverage the proven Lipschitz constant of the parametric BoxQP, incorporating that of the Koopman observable, to certify the robustness of Koopman-MPC solutions \citep{teichrib2023efficient}.

\begin{remark}[\textit{Feasibility-guaranteed}]
    The dynamics-relaxed Koopman-BoxQP \eqref{eqn_KoopmanMPC_BoxQP} (or \eqref{eqn_BoxQP}) always admits a unique solution due to convexity and non-empty constraints, thus guaranteeing feasibility.
\end{remark}

\begin{lemma}[\textit{Lipschitz-guaranteed}]\label{lemma_Lipschitz}
Assume that the lifted mapping $\psi(\cdot)$ is Lipschitz continuous with constant $L_{\psi}$. The feedback policy of Koopman-BoxQP \eqref{eqn_KoopmanMPC_BoxQP},
    \begin{equation}
        u_0(x(t))\triangleq C_{\text{policy}}z^*(x(t)),
    \end{equation}
    where $C_{\text{policy}}=[I_{n_u},0]$), is Lipschitz continuous with constant $\frac{ \rho\sqrt{\lambda_{\max}(\mathbf{E}^\top(\mathbf{F}\mathbf{F}^\top+I)\mathbf{E})}L_{\psi}}{\lambda_{\min}(H)}$.
\end{lemma}
\begin{proof}
    The optimal solution $z^*(x(t))$ can be characterized by the variational inequality: 
    \begin{equation}\label{eqn_VI}
          \big(Hz^*(x(t))+h(x(t))\big)^\top\big(z-z^*(x(t))\big), \ \forall z\in [-\mathbf{1},\mathbf{1}].  
    \end{equation}
    Let us take $x(t_1),x(t_2)$ with corresponding solutions $z_1\triangleq z^*(x(t_1))$,  $z_2\triangleq z^*(x(t_2))$ and define $\Delta h\triangleq h(x(t_2))-h(x(t_1))$, $\Delta z\triangleq z_2 - z_1$, $\Delta x\triangleq x(t_2)-x(t_1)$. Applying \eqref{eqn_VI} at each solution---$\text{for~} z_1: (Hz_1+h(x(t_1)))^\top(z_2-z_1)\geq0,~\text{for~} z_2: (Hz_2+h(x(t_2)))^\top(z_1-z_2)\geq0$---and then adding the two inequalities results in 
\[
(Hz_1+h(x(t_1)))^\top(z_2-z_1) + (Hz_2+h(x(t_2)))^\top(z_1-z_2)\geq0,
\]
that is, $\Delta z^\top H \Delta z\leq -\Delta z^\top \Delta h$.
$H\succ0$ implies the strong convexity inequality $\lambda_{\min}(H)\|\Delta z\|_2^2 \leq \Delta z^\top H \Delta z$. By the Cauchy-Schwarz inequality and the assumption, $| \Delta z^\top \Delta h| \leq \| \Delta z\|_2 \|\Delta h\|_2$ and
\[
    \| \Delta h\|_2\leq \rho \|[\mathbf{F}^\top \mathbf{E};-\mathbf{E}]\|_2 L_{\psi}\|\Delta x\|_2=\rho \sqrt{\lambda_{\max}(\mathbf{E}^\top(\mathbf{F}\mathbf{F}^\top+I)\mathbf{E})}L_{\psi}\|\Delta x\|_2.
\]
Combining those inequalities gives that
\[
\lambda_{\min}(H)\|\Delta z\|_2^2 \leq   \| \Delta z\|_2  \rho \sqrt{\lambda_{\max}(\mathbf{E}^\top(\mathbf{F}\mathbf{F}^\top+I)\mathbf{E})}L_{\psi}\|\Delta x\|_2.
\]
That is, $\|\Delta z\|_2\leq\frac{ \rho \sqrt{\lambda_{\max}(\mathbf{E}^\top(\mathbf{F}\mathbf{F}^\top+I)\mathbf{E})}L_{\psi}}{\lambda_{\min}(H)}\|\Delta x\|_2$, namely,
\[
\quad\|u_0(x(t_2))-u_0(x(t_1))\|_2 \leq \|C_{\text{policy}}\| \|\Delta z\|_2\leq\frac{ \rho\sqrt{\lambda_{\max}(\mathbf{E}^\top(\mathbf{F}\mathbf{F}^\top+I)\mathbf{E})}L_{\psi}}{\lambda_{\min}(H)}\|x(t_2)-x(t_1)\|_2,
\]
which completes the proof.
\end{proof}

\newpage

\subsection{Feasible Mehrotra's IPM Algorithm for BoxQP}\label{sec_IPM_algorithm}
\begin{wrapfigure}[12]{r}{0.4\textwidth} 
\vspace{-10pt}
\rule{\linewidth}{0.6pt}
\textbf{KKT condition of BoxQP}
\begin{subequations}\label{eqn_KKT}
\begin{align}
    Hz + h(x(t)) + \gamma - \theta = 0,\label{eqn_KKT_a}\\
    z + \phi - \mathbf{1}_n=0,\label{eqn_KKT_b}\\
    z - \psi + \mathbf{1}_n=0,\label{eqn_KKT_c}\\
    (\gamma,\theta,\phi,\psi)\geq0,\label{eqn_KKT_d}\\
    \gamma \pdot \phi = 0,\label{eqn_KKT_e}\\
    \theta \pdot \psi = 0,\label{eqn_KKT_f}
\end{align}
\end{subequations}   
\rule{\linewidth}{0.6pt}
\end{wrapfigure}
According to \citet[Ch 5]{boyd2004convex}, the Karush–Kuhn–Tucker (KKT) condition of the BoxQP \eqref{eqn_BoxQP} is \eqref{eqn_KKT}, where $\gamma,\theta$ are the Lagrangian variables of the lower and upper bound, respectively; $\phi,\psi$ are the slack variables of the lower and upper bound, respectively; $\pdot$ is the Hadamard product, i.e., $\gamma\pdot \phi = \mathrm{col}(\gamma_1\phi_1,\gamma_2\phi_2,\cdots{},\gamma_n\phi_n)$. The BoxQP \eqref{eqn_BoxQP} may be ill-conditioned since the penalty parameter $\rho$ needs to be large to ensure strong closed-loop performance. To address this issue, this article employs second-order IPMs, with a particular focus on the most computationally efficient variant: Mehrotra’s predictor–corrector IPM \citep{mehrotra1992implementation}. In practice, Mehrotra’s predictor–corrector IPM typically achieves highly accurate solutions within $50\sim75$ iterations, regardless of whether the initial point is strictly feasible. This remarkably fast convergence makes the method the foundation of most IPM software packages.

\textit{An important numerical observation in this article} is that Mehrotra’s predictor–corrector IPM typically takes \textit{about 10 iterations} to find a highly accurate solution of BoxQP \eqref{eqn_BoxQP} when applying the following cost-free initialization strategy. To demonstrate this, denote the feasible region by $\mathcal{F}$, i.e., $\mathcal{F}=\{(z,\gamma,\theta,\phi,\psi):\eqref{eqn_KKT_a}\mathrm{-}\eqref{eqn_KKT_c},(\gamma,\theta,\phi,\psi)\geq0\}$ and the set of strictly feasible points by $    \mathcal{F}^+\triangleq\{(z,\gamma,\theta,\phi,\psi):\eqref{eqn_KKT_a}\mathrm{-}\eqref{eqn_KKT_c},(\gamma,\theta,\phi,\psi)>0\}$. For a point $\in\mathcal{F}^+$, the residuals of  \eqref{eqn_KKT_a}--\eqref{eqn_KKT_c} are zeros and only the residuals of the complementary equations \eqref{eqn_KKT_e}--\eqref{eqn_KKT_f} are not.

\subsection{Cost-free strictly feasible cold- and warm-starting initializations}
IPMs for general optimization problems, including QPs, generally do not support warm-start initialization, which is an unfortunate limitation as warm-start initialization can substantially accelerate convergence and reduce iteration counts in real-time MPC applications. A notable advantage of the proposed \textit{Koopman-BoxQP} \eqref{eqn_KoopmanMPC_BoxQP} is its ability to support strictly feasible cold- and warm-start initializations without any extra computational overhead.
\begin{remark}\label{remark_initialization}
(\textbf{Strictly feasible cold-start}): Inspired from \citep{wu2025direct,wu2025quadratic}, the initial point
\begin{equation}\label{eqn_cold_initialization_stragegy}
\footnotesize
    (z^0, \gamma^0, \theta^0, \phi^0,\psi^0)   = (0,\|h(x(t))\|_\infty\mathbf{1}_n - \tfrac{1}{2}h(x(t)),\|h(x(t))\|_\infty\mathbf{1}_n + \tfrac{1}{2}h(x(t)),\mathbf{1}_n,\mathbf{1}_n)\in\mathcal{F}^+.
\end{equation}
\end{remark}

\begin{remark}
(\textbf{Strictly feasible warm-start}): 
The optimal solution at the previous sampling time $t-1$: $ U^{t-1}=\mathrm{col}(u_0^{t-1},\cdots{},u_{N-1}^{t-1}),~X^{t-1}=\mathrm{col}(x_1^{t-1},\cdots{},x_{N}^{t-1})$, can be shifted ahead one step to obtain a warm start initial point: $z^{t,\text{guess}}\triangleq\mathrm{col}(u_1^{t-1},\cdots{},u_{N-1}^{t-1},u_{N-1}^{t-1},x_2^{t-1},\cdots{},x_{N}^{t-1},x_{N}^{t-1})$ (Clearly, $-\mathbf{1}_n\leq  z^{t,\text{guess}}\leq\mathbf{1}_n$) to be used in solving the Koopman-BoxQP \eqref{eqn_KoopmanMPC_BoxQP} at the current sampling time $t$. The warm start initial point is 
\begin{equation}\label{eqn_warm_initialization_stragegy}
    z^0 = z^{t,\text{guess}},~
    \gamma^0 =\|\bar{h}\|_{\infty}\mathbf{1}_n - \textbf{}\tfrac{1}{2}\bar{h},~
    \theta^0 =\|\bar{h}\|_{\infty}\mathbf{1}_n + \tfrac{1}{2}\bar{h},~
    \phi^0 = \mathbf{1}_n -  z^{t,\text{guess}},~
    \psi^0 = \mathbf{1}_n +  z^{t,\text{guess}},
\end{equation}
where $\bar{h}\triangleq H z^{t,\text{guess}}+h(x(t))$. Notably, we have $(z^0,\gamma^0,\theta^0,\phi^0,\psi^0)\in\mathcal{F}^+$.
\end{remark}

\newpage
\subsection{Algorithm description and efficient computation of Newton systems}
\begin{wrapfigure}[35]{r}{0.4\textwidth}
\vspace{-20pt}
\captionsetup{type=algorithm}
\begin{minipage}{0.5\textwidth}
\begin{algorithm}[H]
    \caption{Feasible Mehrotra's predictor-corrector IPM for BoxQP \eqref{eqn_BoxQP}} \label{alg_IPM}
    {\scriptsize
    \textbf{Input}: Initializing $(z^0,\gamma^0,\theta^0,\phi^0,\psi^0)$ from Eqn. \eqref{eqn_cold_initialization_stragegy}, a desired optimal level $\epsilon$, and the maximum number of iterations $N_{\max}$.
    \vspace*{.1cm}\hrule\vspace*{.1cm}

    \textbf{for} $k=0,1, 2,\cdots{},N_{\max}-1$ \textbf{do}
    \begin{enumerate}[label*=\arabic*., ref=\theenumi{}]
        \item $\mu^k \leftarrow[(\gamma^k)^\top \phi^k + (\theta^k)^\top \psi^k]/(2n)$;
        \item if $\mu^k\leq (2n)\epsilon$, then break;
        \item Compute $(\Delta z^{\text{aff}},\Delta \gamma^{\text{aff}},\Delta \theta^{\text{aff}},\Delta \phi^{\text{aff}},\Delta \psi^{\text{aff}})$ by solving
        \begin{equation}\label{eqn_predictor_step}
        \scriptsize
           J^k \left[\begin{array}{c}
                 \Delta z^{\text{aff}}\\
                 \Delta \gamma^{\text{aff}}\\
                 \Delta \theta^{\text{aff}}\\
                 \Delta \phi^{\text{aff}}  \\
                 \Delta \psi^{\text{aff}}
            \end{array}\right] = \left[\begin{array}{c}
                 0  \\
                 0 \\
                 0 \\
                 -\gamma^k \pdot \phi^k\\
                 -\theta^k \pdot \psi^k
            \end{array}\right]
        \end{equation}
        where 
        \begin{equation}
        \scriptsize
            J^k\triangleq\left[\begin{array}{ccccc}
                H & I & -I & 0 & 0 \\
                I & 0 & 0 & I & 0 \\
                I & 0 & 0 & 0 & -I\\
                0 & D(\phi^k) & 0 & D(\gamma^k) & 0\\
                 0 & 0 & D(\psi^k)  & 0 & D(\theta^k) \\
            \end{array}\right]
        \end{equation}
        \item $\alpha^{\text{aff}}=\min(1,0.99\min_{\Delta \gamma_i^{\text{aff}}<0}\frac{-\gamma^k_i}{\Delta \gamma_i^{\text{aff}}},0.99\min_{\Delta \theta_i^{\text{aff}}<0}\frac{-\theta^k_i}{\Delta \theta_i^{\text{aff}}}$ $,0.99\min_{\Delta \phi_i^{\text{aff}}<0}\frac{-\phi^k_i}{\Delta \phi_i^{\text{aff}}},0.99\min_{\Delta \psi_i^{\text{aff}}<0}\frac{-\psi^k_i}{\Delta \psi_i^{\text{aff}}})$;
        \item $\mu^{\text{aff}}=[(\gamma^k+\alpha^{\text{aff}}\Delta \gamma^{\text{aff}})^\top (\phi^k+\alpha^{\text{aff}}\Delta \phi^{\text{aff}})$$+(\theta^k+\alpha^{\text{aff}}\Delta \theta^{\text{aff}})^\top (\psi^k+\alpha^{\text{aff}}\Delta \psi^{\text{aff}})]/(2n)$;
        \item $\sigma\leftarrow \left(\frac{\mu^{\text{aff}}}{\mu^k}\right)^3$;
        \item Compute $(\Delta z,\Delta \gamma,\Delta \theta,\Delta \phi,\Delta \psi)$ by solving
          \begin{equation}\label{eqn_corrector_step}
          \tiny
           J^k \left[\begin{array}{c}
                 \Delta z\\
                 \Delta \gamma\\
                 \Delta \theta\\
                 \Delta \phi \\
                 \Delta \psi
            \end{array}\right] = \left[\begin{array}{c}
                 0  \\
                 0 \\
                 0 \\
                 -\gamma^k \pdot \phi^k-\Delta\gamma^{\text{aff}}\pdot\Delta\phi^{\text{aff}}+\sigma\mu^k\mathbf{1}_n\\
                 -\theta^k \pdot \psi^k-\Delta\theta^{\text{aff}}\pdot\Delta\psi^{\text{aff}}+\sigma\mu^k\mathbf{1}_n
            \end{array}\right]
        \end{equation}
        \item $\alpha=\min(1,0.99\min_{\Delta \gamma_i<0}\frac{-\gamma^k_i}{\Delta \gamma_i},0.99\min_{\Delta \theta_i<0}\frac{-\theta^k_i}{\Delta \theta_i}$ $,0.99\min_{\Delta \phi_i<0}\frac{-\phi^k_i}{\Delta \phi_i},0.99\min_{\Delta \psi_i<0}\frac{-\psi^k_i}{\Delta \psi_i})$;
        \item $z^{k+1}\leftarrow z^k+\alpha\Delta z,\leftarrow \gamma^k+\alpha\Delta \gamma,~\leftarrow \theta^k+\alpha\Delta \theta,\leftarrow \phi^k+\alpha\Delta \phi,~\psi^{k+1}\leftarrow \psi^k+\alpha\Delta \psi$;
    \end{enumerate}
    \textbf{end}\vspace*{.1cm}\hrule\vspace*{.1cm}
    \textbf{Output:} $z^{k+1}$.
    }
\end{algorithm}
\end{minipage}
\end{wrapfigure}
Algorithm \ref{alg_IPM} presents a \textit{feasible} variant of Mehrotra’s predictor–corrector IPM, specifically tailored for BoxQP \eqref{eqn_BoxQP}. At each iteration, Steps 3 and 4 of Algorithm \ref{alg_IPM} compute the predictor and corrector directions, respectively, using the same coefficient matrix $J_k$; thus, only a single matrix factorization is required. Moreover, the structured sparsity of $J_k$ enables a more computationally efficient Cholesky factorization on a reduced-dimensional linear system, from $5N\times(n_u+n_x)\rightarrow N\times n_u$. Considering \eqref{eqn_predictor_step} and \eqref{eqn_corrector_step}, which differ only in the last two terms on the right-hand side, denote these terms as $r^1$ and $r^2$, respectively. Then, by letting
\begin{equation}\label{eqn_Delta_gamma_theta_phi_psi}
\begin{aligned}
        \Delta \gamma &=  \frac{\gamma^k}{\phi^k}\Delta z + \frac{r^1}{\phi^k},~\Delta \theta = -\frac{\theta^k}{\psi^k}\Delta z+\frac{r^2}{\psi^k},\\
        \Delta\phi &= - \Delta z,~ \Delta\psi = \Delta z,       
\end{aligned} 
\end{equation}
Equation \eqref{eqn_predictor_step} (or \eqref{eqn_corrector_step}) can be reduced to a more compact system of linear equations,
\begin{equation}\label{eqn_compact_linsys}
 \bar{H}\Delta z=\frac{r^1}{\phi} - \frac{r^2}{\psi},
\end{equation}
where the coefficient matrix $\bar{H}$ has a $2 \times 2$ block structure from the definition of $H$ in \eqref{eqn_H_def}:
\begin{equation}
    \bar{H}\triangleq H+D\!\left(\frac{\gamma^k}{\phi^k}+\frac{\theta^k}{\psi^k}\right) =\left[\begin{array}{@{}cc@{}}
        \bar{H}_{11} &  -\rho \mathbf{F}^\top\\
        \textcolor{blue}{-}\rho \mathbf{F} & \bar{H}_{22}
    \end{array}\right]\!\succ0
\end{equation}
(where $D(\cdot)$ denotes the diagonal matrix of a vector) with
\[
\begin{aligned}
    \bar{H}_{11}&\triangleq\rho\mathbf{F}^\top\mathbf{F}+\bar{W}_u+\bar{R}+D\!\left(\frac{\gamma_{1:n_1}^k}{\phi_{1:n_1}^k}+\frac{\theta_{1:n_1}^k}{\psi_{1:n_1}^k}\right)\\
    \bar{H}_{22}&\triangleq \rho I +\bar{W}_x+D\!\left(\frac{\gamma_{n_1+1:n}^k}{\phi_{n_1+1:n}^k}+\frac{\theta_{n_1+1:n}^k}{\psi_{n_1+1:n}^k}\right)
\end{aligned}
\]
where $\bar{H}_{11}\succ0\in\rr^{n_1\times n_1}$, $\bar{H}_{22}\succ0\in\rr^{n_2\times n_2}$, $n_1=N n_u$, and $n_2=N n_x$. By assumption the weighting matrix $W_x$ is diagonal, thus $\bar{W}_x$ and $\bar{H}_{22}$ are also diagonal. By exploiting the diagonal structure of $\bar{H}_{22}$, computational savings in solving  \eqref{eqn_compact_linsys} can be achieved through solving the system:
\begin{equation}\label{eqn_efficient_linsys}
\left(\bar{H}_{11}-\rho^2 \mathbf{F}^\top \bar{H}_{22}^{-1}\mathbf{F}\right)\! \Delta z_{1:n_1} 
   = \frac{r^1_{1:n_1}}{\phi_{1:n_1}^k} - \frac{r^2_{1:n_1}}{\psi_{1:n_1}^k} + \rho \mathbf{F}^\top\bar{H}_{22}^{-1}\!\left(\frac{r^1_{n_1+1:n}}{\phi_{n_1+1:n}^k} - \frac{r^2_{n_1+1:n}}{\psi_{n_1+1:n}^k} \right)    
\end{equation}
and $\Delta z_{n_1+1:n} = \bar{H}_{22}^{-1}\!\left(\frac{r^1_{n_1+1:n}}{\phi_{n_1+1:n}^k} - \frac{r^2_{n_1+1:n}}{\psi_{n_1+1:n}^k} + \rho \mathbf{F}\Delta z_{1:n_1} \right)$ to obtain the solution $\Delta z=\mathrm{col}(\Delta z_{1:n_1},\Delta z_{n_1+1:n})$.
\begin{remark}
By the Schur Complement lemma, $\bar{H}_{11}-\rho^2 \mathbf{F}^\top \bar{H}_{22}^{-1}\mathbf{F}\succ0$, and the matrix is symmetric. Then, the Cholesky factorization with the cost $O((N n_u)^3)$ for \eqref{eqn_efficient_linsys} can be applied, which is crucial for achieving computational speedup.
\end{remark}

\section{Numerical Examples}
\subsection{Practical behavior of cold-starting Mehrotra's IPM algorithm on random BoxQPs}
\begin{wrapfigure}[17]{r}{0.5\textwidth}
\vspace{-10pt}
    \centering
    \includegraphics[width=1\linewidth]{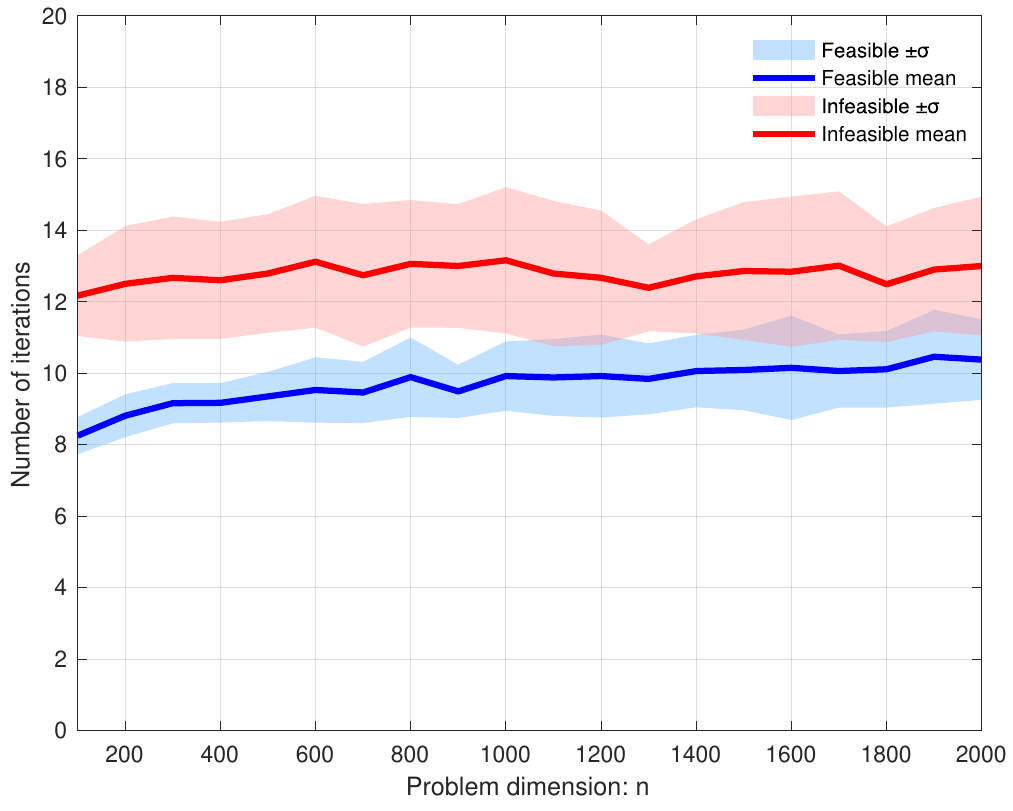}
    \caption{Iteration counts of feasible and infeasible Mehrotra's IPM algorithms on ill-conditioned random BoxQP problems with dimensions ranging from $100$ to $2000$.}\label{fig_iters_feasible_infeasible}
\end{wrapfigure}
The main inspiration for using the Koopman-BoxQP formulation comes from our contributed numerical observation as follows,
\begin{quote}
\textit{Feasible Mehrotra's IPM algorithm usually takes $\sim 10$ iterations on (ill-conditioned) random BoxQP problems, regardless of the problem dimension.}
\end{quote}
To verify this observation, we generate random BoxQP problems with dimensions $n$ ranging from $100$ to $2000$, where the Hessian matrix has a condition number of $10^6$. Algorithm \ref{alg_IPM}, with the desired optimality tolerance set $\epsilon=10^{-6}$, is applied both with and without the cold-start initialization strategy, corresponding to the feasible and infeasible variants, respectively, for solving the generated BoxQP instances. For each specified problem dimension, the BoxQP is randomly generated and solved 100 times to characterize the statistical behavior of the iteration counts of feasible and infeasible Mehrotra's IPM algorithms, which are plotted in Fig.\ \ref{fig_iters_feasible_infeasible}.  Clearly, Fig.\ \ref{fig_iters_feasible_infeasible} confirms our key numerical observation: the feasible Mehrotra’s IPM algorithm (Algorithm \ref{alg_IPM}) is \textbf{scalable}, requiring only about 10 iterations to achieve highly accurate solutions.

\subsection{Performance validation of \textit{Koopman-BoxQP} on PDE NMPC problems}
This section applies our proposed \textit{Koopman-BoxQP} to solve a PDE-MPC problem with state and control input constraints, which is a large-sized QP problem with $1040$ variables and $2080$ constraints. The PDE plant under consideration is the nonlinear Korteweg-de Vries (KdV) equation that models the propagation of acoustic waves in plasma or shallow water waves \citep{miura1976korteweg} as
\begin{equation}\label{eqn_KdV}
\frac{\partial y(t, x)}{\partial t}+y(t, x) \frac{\partial y(t, x)}{\partial x}+\frac{\partial^3 y(t, x)}{\partial x^3}=u(t,x)
\end{equation}
where $x\in[-\pi,\pi]$ is the spatial variable. Consider the control input $u(t,x)=\sum_{i=1}^4u_i(t)v_i(x)$, in which the four coefficients $\{u_i(t)\}$ are subject to the constraint $[-1,1]$, and $v_i(x)$ are predetermined spatial profiles given as $v_i(x)=e^{-25(x-m_i)^2}$, with $m_1=-\pi/2$, $m_2=-\pi/6$, $m_3=\pi/6$, and  $m_4=\pi/2$. The control objective is to adjust ${u_i(t)}$ so that the spatial profile $y(t,x)$ tracks the given reference signal. The spatial profile $y(t,x)$ is uniformly discretized into $100$ spatial nodes. These 100 discretized nodes serve as the system states, each constrained within $[-1,1]$, while the control inputs are likewise bounded within $[-1,1]$. Choosing the MPC prediction horizon as $N=10$ results in a medium-scale optimization problem with $1040$ variables and $2080$ constraints. For the MPC settings, the state cost matrix is set to $W_x=I_{100}$, and the control inputs matrix is set to $W_u=0.05 I_{4}$. The state references $x_r\in\mathbb{R}^{100}$ are sinusoidal signals for a $50$ s simulation time, and the control input reference is a constant value with $u_r=0$. We employ a spectral method based on the Fourier transform and a split-step scheme to solve the nonlinear KdV \eqref{eqn_KdV}, for data generation and MPC closed-loop simulation. The sampling time is chosen to be $\Delta t=0.01$. 

\begin{wrapfigure}[28]{r}{0.5\textwidth} 
\vspace{-10pt}
\begin{minipage}[b]{0.5\textwidth}
\centering
\includegraphics[width=\linewidth]{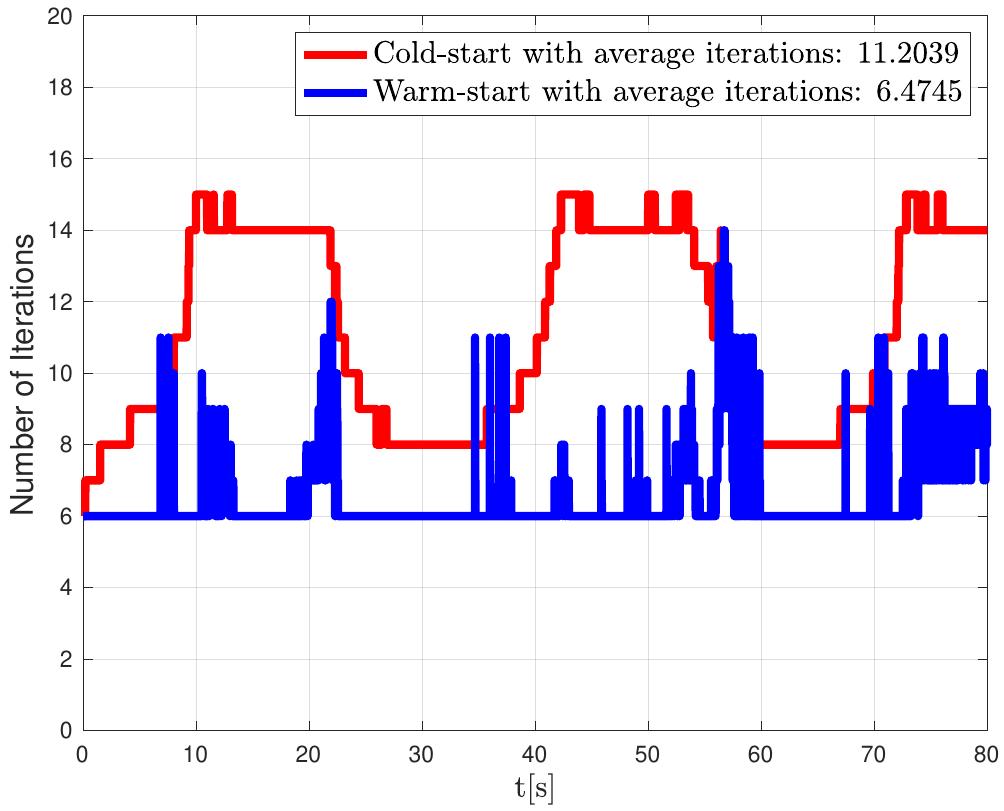}
\caption{Iterations comparison of Algorithm \ref{alg_IPM} with Cold-start and Warm-start.}
\label{fig_WarmStart_ColdStart}
\end{minipage}
\hfill
\begin{minipage}[b]{0.5\textwidth}
\centering
\vspace{8pt} 
\begin{tabular}{lcc}
\toprule
\multirow{2}{*}{QP Solver}  & \multicolumn{2}{c}{Average Execution Time [s]\footnotemark} \\ 
& cold start & warm start\\
\midrule
Quadprog & 0.1127 &0.1127 \\
OSQP     & $9.7$$\times$$10^{-3}$  & $1.2$$\times$$10^{-3}$ \\
SCS      & $5.8$$\times$$10^{-3}$ & $5.8$$\times$$ 10^{-3}$\\
Algorithm \ref{alg_IPM} & $\textcolor{blue}{\mathbf{5.57\times10^{-4}}}$ &  $\textcolor{blue}{\mathbf{3.63\times10^{-4}}}$ \\ 
\bottomrule
\end{tabular}
\captionof{table}{Execution time comparison between Algorithm \ref{alg_IPM} and other state-of-the-art QP solvers.}
\label{tab_execution_time}
\end{minipage}
\end{wrapfigure}
\footnotetext{The execution time results are based on MATLAB's C-MEX implementation of Algorithm \ref{alg_IPM} running on a Mac mini with an Apple M4 Chip (10-core CPU and 16 GB RAM).} 
The setting for our closed-loop simulation includes: \textit{(i) Data generation:} The data are collected from 1000 simulation trajectories with 200 samples. At each simulation, the initial condition of the spatial profile is a random combination of four given spatial profiles, i.e., $y_1^0(0,x)=e^{-(x-\pi/2)^2}$, $y_2^0(0,x)=-\operatorname{sin}(x/2)^2$, $y_3^0(0,x)=e^{-(x+\pi/2)^2}$, $y_4^0(0,x)=\operatorname{cos}(x/2)^2$. The four control inputs $u_i(t)$ are distributed uniformly in $[-1,1]$; \textit{(ii) Koopman predictor:} The lift function $\psi$ is composed of the original 100 spatial states and 200 thin-plate RBFs with random centers, which leads to the lifted state dimension $N_{\psi}=100+ 200=300$. Then the lifted linear predictor with $A\in\mathbb{R}^{300\times 300}$ and $B\in\mathbb{R}^{300\times 4}$ is obtained from the Moore-Penrose pseudoinverse of the lifting data matrix, and its output matrix is $C=[I_{100},0]\in\mathbb{R}^{100\times 300}$; \textit{(iii) \textit{Koopman-BoxQP} formulation:} Given the matrices $\{A,B,C\}$ calculated from the above previous Koopman predictor step, the approximated multi-step Koopman model \eqref{eqn_X_U_E_F} is constructed and embedded into the \textit{Koopman-BoxQP} \eqref{eqn_KoopmanMPC_BoxQP} with dynamic penalty parameter $\rho=10^2$. This results in a BoxQP with $1040$ variables and $2080$ constraints.

The proposed Algorithm \ref{alg_IPM} with the stopping criteria $\epsilon=1$$\times$$10^{-6}$ is applied to solve the resulting large-size \textit{Koopman-BoxQP} problem \eqref{eqn_KoopmanMPC_BoxQP} in the closed-loop simulation. The closed-loop performance, shown in Fig.\ \ref{fig_kdv}, demonstrates that the spatial profile $y(t,x)$ accurately tracks the desired reference signals while remaining within the state constraint $[-1,1]$, even when the reference signals exceed this bound. The control inputs also remain strictly within their limits of $[-1,1]$. This numerical case study demonstrates that the proposed \textit{Koopman-BoxQP} approach can accurately control the evolution of large-scale states.
\begin{figure}[!htbp]\label{fig}
\begin{picture}(140,110)
\put(0,-12){\includegraphics[width=55mm]{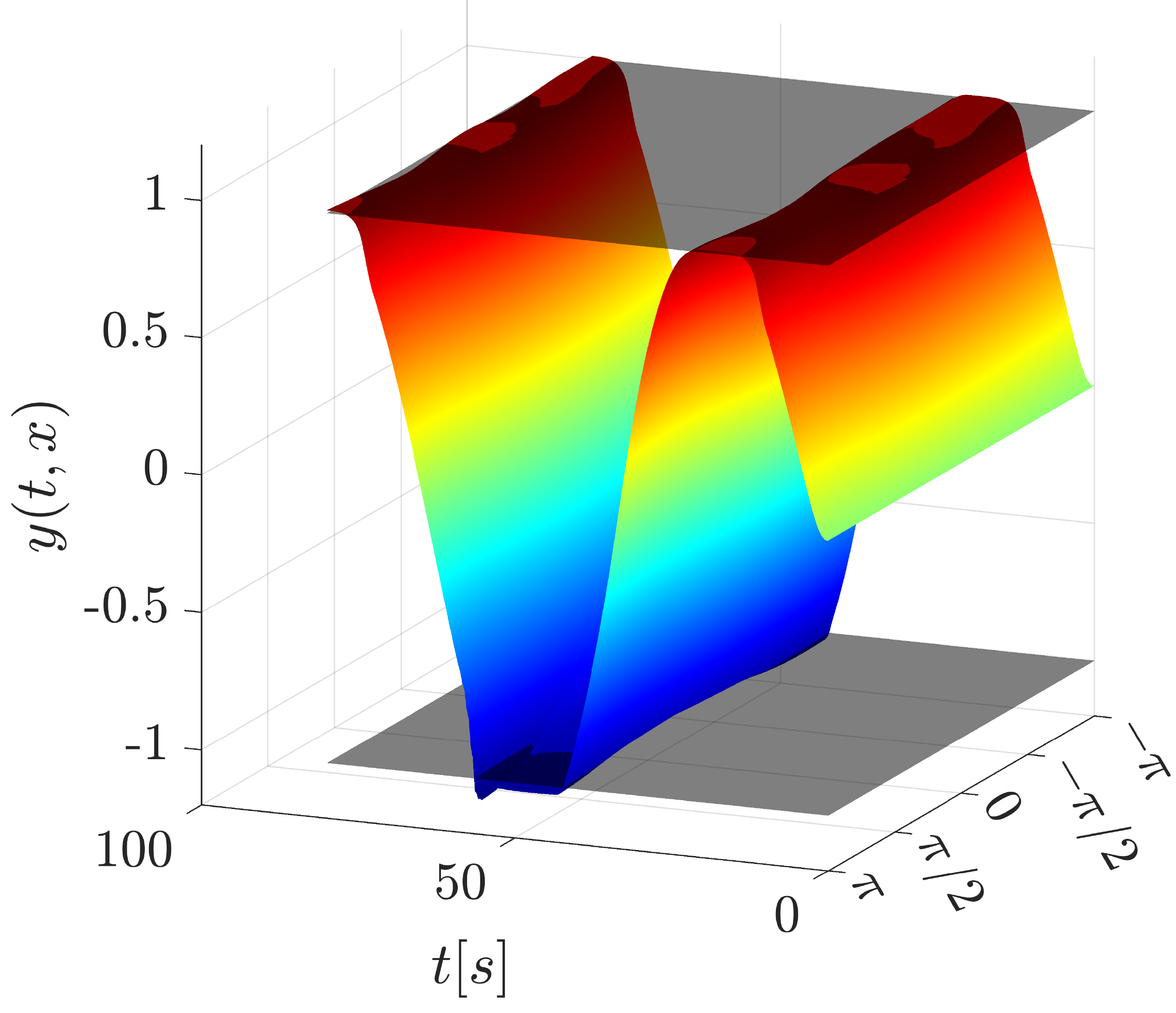}}
\put(160,-10){\includegraphics[width=45mm]{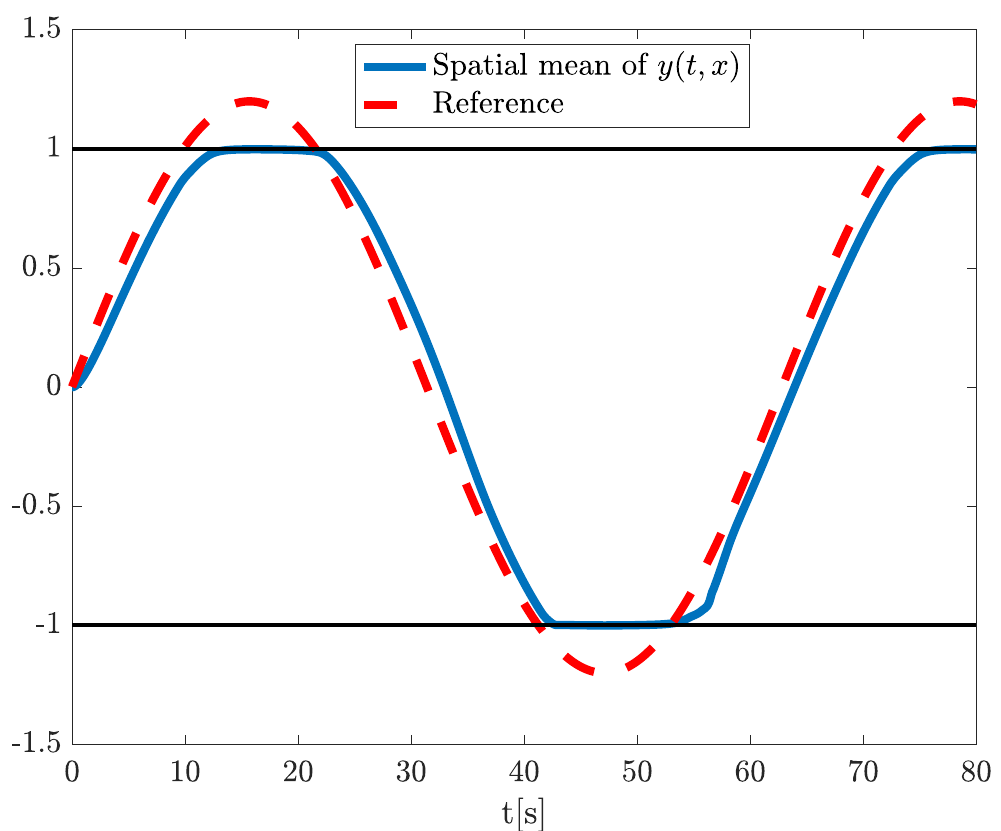}}
\put(300,-10){\includegraphics[width=45mm]{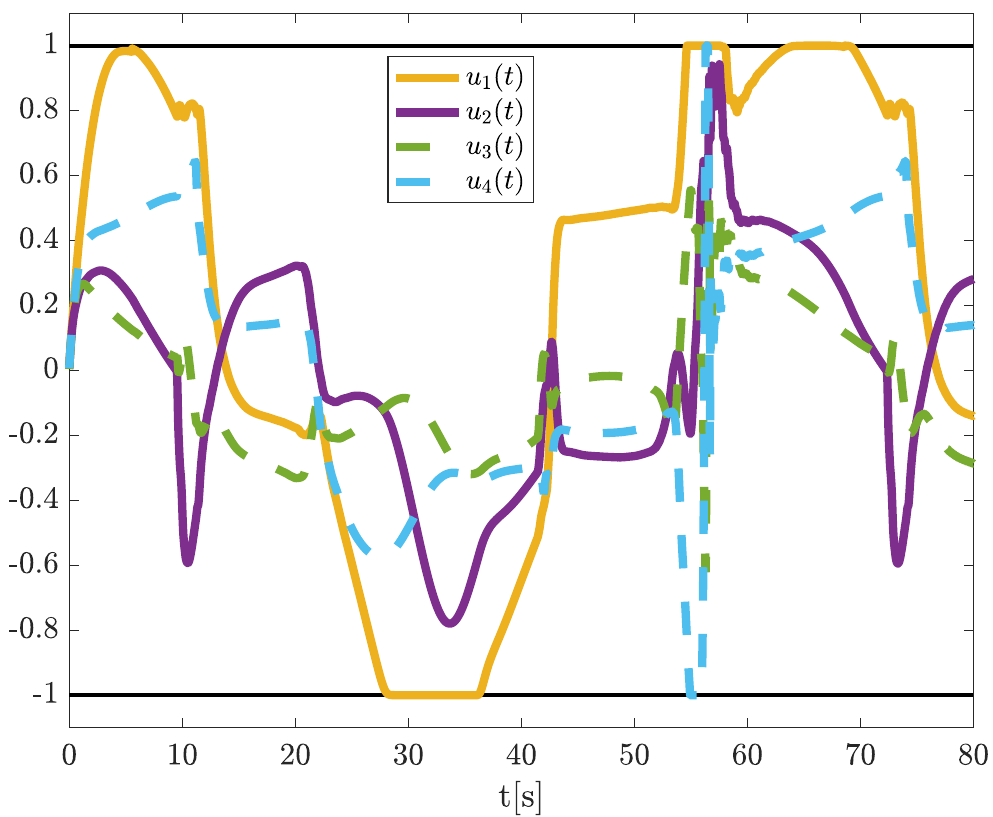}}
\end{picture}
\caption{Closed-loop simulation of the nonlinear KdV system with the dynamics-relaxed \textit{Koopman-BoxQP} controller tracking a time-varying spatial profile reference. Left: time evolution of the spatial profile $y(t,x)$ and the state constraints $[-1,1]$. Middle: spatial mean of the $y(t,x)$ and the state constraints $[-1,1]$. Right: the four control inputs and the control input constraints $[-1,1]$.}
\label{fig_kdv}
\end{figure}

The number of iterations of the IPM-based Algorithm \ref{alg_IPM} with cold and warm starts is plotted in Fig.\  \ref{fig_WarmStart_ColdStart}, which demonstrates that Algorithm \ref{alg_IPM} effectively supports warm-start initialization, reducing the average number of iterations by half, from $\mathbf{11.2039}$ (cold start) to $\mathbf{6.4745}$ (warm start). This implies that solving around $\mathbf{6}$ small-scale linear systems (reduced from large-scale via exploiting the problem structure) can obtain $10^{-6}$-accuracy solutions, \textit{explaining why Algorithm~\ref{alg_IPM} can solve large-scale NMPC problems at kHz rates (in milliseconds)} as shown in Table \ref{tab_execution_time}. Other state-of-the-art QP solvers, including MATLAB’s Quadprog (using IPM), OSQP \citep{stellato2020osqp}, and SCS \citep{odonoghue:21} (all choose the `eps\_rel', `1e-6' for a fair comparison), are compared for solving the same resulting \textit{Koopman-BoxQP} \eqref{eqn_KoopmanMPC_BoxQP}. Their average execution times under cold- and warm-start initialization are listed in Table \ref{tab_execution_time}, which shows that using a warm start is not helpful in the IPM-based solvers Quadprog and SCS, but beneficial in Algorithm \ref{alg_IPM}. Furthermore, Algorithm \ref{alg_IPM} is the only solver that achieves kHz-level solution speeds for large-scale NMPC problems. 
Moreover, the proposed Algorithm \ref{alg_IPM} has been extended to general QPs, and more numerical experiments can be found in the Appendix \ref{sec_appendix} due to space limits.

\section{Conclusion}
This article presents a fast \textit{Koopman-BoxQP} framework enabling NMPC at kHz rates. Our approach consists of a \textit{dynamics-relaxed} BoxQP formulation and a \textit{linear-algebra-tailored} and \textit{warm-starting-supported} IPM-based QP solver. We validated the approach on a challenging large-sized PDE-MPC problem (with 1040 decision variables and 2080 inequality constraints) and achieved a real-time solution under $<1$ ms on a standard desktop CPU, opening a new era of kHz-rate solutions for large-scale NMPC. Future work will focus on developing a GPU-accelerated \textit{Koopman-BoxQP} solver, leveraging the computational power of cuBLAS and cuSOLVER for large-scale linear system solves, together with the practical scalable iteration complexity (typically around $10$ iterations) of the IPM in BoxQPs, to enable larger PDE-MPC applications. The stability proof and robustness of the \textit{dynamics-relaxed Koopman-BoxQP} will also be investigated.

\acks{
 This research was supported by the Ralph O’Connor Sustainable Energy Institute at Johns Hopkins University. Wallace Tan was supported by the MathWorks Fellowship. Richard Braatz was supported by the U.S. Food and Drug Administration under the FDA BAA-22-00123 program (Award 75F40122C00200).}

\bibliography{l4dc2026-sample}

\appendix
\section{Numerical experiments on general QPs}\label{sec_appendix}
The proposed BoxQP solver in the dynamics-relaxed KoopmanMPC-to-BoxQP approach can be extended to general QPs, not limited to nonlinear MPC applications. 
Consider a general convex QP,
\begin{subequations}\label{eqn_bound_QP}
    \begin{align}
        \min_x&~ \tfrac{1}{2}x^\top Q x + q^\top x\\
        \text{s.t. }&~ y_{\min}\leq Ax\leq y_{\max}\label{eqn_bound_soft}  \\
        &~ x_{\min} \leq x \leq x_{\max} \label{eqn_bound_hard}   
    \end{align}
\end{subequations}
where $x\in\rr^{n_x}$, $Q=Q^\top\succeq0$ and $Q\in\rr^{n_x\times n_x}$, $q\in\rr^{n_x}$, and $A\in\rr^{n_y \times n_x}$. The bound constraint \eqref{eqn_bound_hard} often comes from the physical actuator limits, and  \eqref{eqn_bound_soft} comes from user-specified constraints, such as the states/outputs constraints. To ensure feasibility without violating hard actuator limits, the core idea of this article is the use of the soft-constrained formulation as follows,
\begin{equation}\label{eqn_soft_BoxQP}
    \begin{aligned}
        \min_{x,y}&~ \tfrac{1}{2}x^\top Q x + q^\top x + \tfrac{\rho}{2}\|Ax-y\|_2^2=\frac{1}{2}\!\left[\begin{array}{@{}c@{}}
             x  \\
             y 
        \end{array}\right]^\top\!\left[\begin{array}{@{}cc@{}}
            Q+\rho A^\top A & -\rho A^\top \\
            -\rho A & \rho I
        \end{array}\right]\left[\begin{array}{@{}c@{}}
             x  \\
             y 
        \end{array}\right]+\left[\begin{array}{@{}c@{}}
             x  \\
             y 
        \end{array}\right]^\top\!\left[\begin{array}{@{}c@{}}
             q  \\
             0
        \end{array}\right]\\
        \text{s.t.}&~ \left[\begin{array}{@{}c@{}}
             x_{\min}  \\
             y_{\min} 
        \end{array}\right]  \leq \left[\begin{array}{@{}c@{}}
             x  \\
             y 
        \end{array}\right] \leq \left[\begin{array}{@{}c@{}}
             x_{\max}  \\
             y_{\max} 
        \end{array}\right]
    \end{aligned}
\end{equation}
where $\rho$ is a large penalty parameter, such as $\rho=10^6$. The soft-constrained BoxQP \eqref{eqn_soft_BoxQP} is always feasible, Lipschitz-guaranteed, and supports the lightening fast tailored solver (from two reasons: \textit{i)} only requiring roughly 10 iterations for a high-accuracy solution, \textit{ii)} allowing the dimension reduction of linear systems) especially when $n_y\gg n_x$, such as from PDE-MPC applications. Algorithm \ref{alg_IPM} is compared against OSQP and SCS solvers (all choose the `eps\_rel', `1e-6') on random QPs \eqref{eqn_bound_QP}, and Fig.\  \ref{fig_comparision} shows that Algorithm \ref{alg_IPM} is fastest 
(at least an order of magnitude faster) and only requires 
roughly $10^{-2}$ s on QPs with $8400$ inequalities.

\begin{figure}[!htbp]
\begin{picture}(140,180)
\put(0,0){\includegraphics[width=75mm]{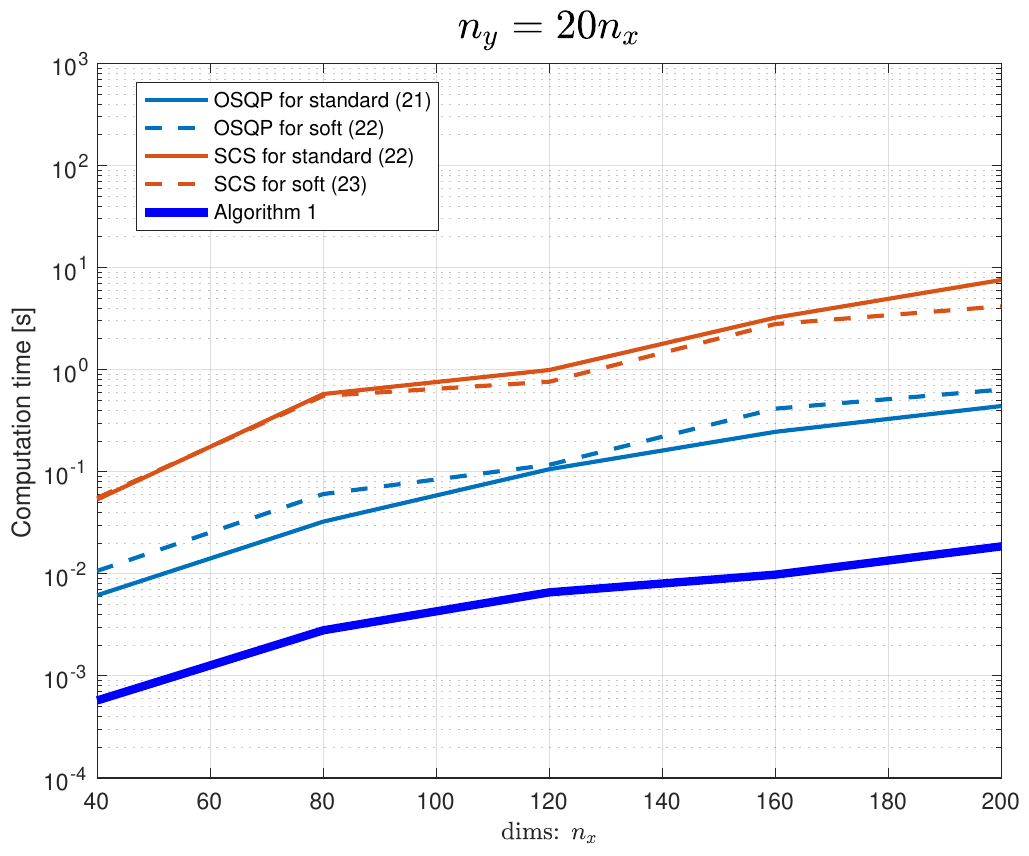}}
\put(225,0){\includegraphics[width=75mm]{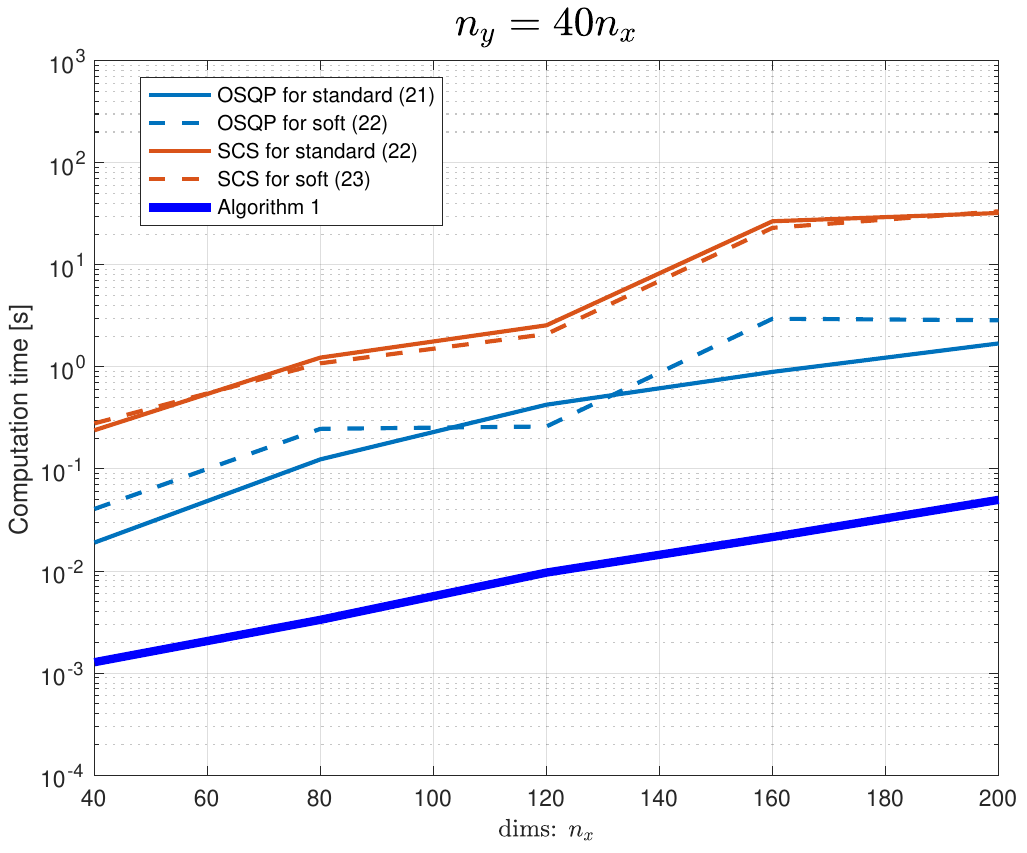}}
\end{picture}

\vspace{-0.3cm}

\caption{Execution time comparison between Algorithm \ref{alg_IPM}, OSQP, and SCS solvers on random QPs (and their soft-constrained QPs with $\rho=10^6$) with different dimensions. Left: $n_y=20 n_x$. Right: $n_y=40 n_x$.}
\label{fig_comparision}
\end{figure}

\end{document}